\documentclass[a4paper,11pt]{article}
\usepackage{indentfirst}
\usepackage{amsmath}
\usepackage{amsthm}
\usepackage{amssymb}
\usepackage{amsfonts}
\usepackage{fancybox}
\usepackage{fancyvrb}
\usepackage{color}
\usepackage{makeidx}
\usepackage{graphicx}
\usepackage{geometry}
\usepackage{array}
\usepackage{cite}
\usepackage{url}
\usepackage{enumerate}
\usepackage{wrapfig}
\usepackage{abstract}
\usepackage{subcaption}

\newtheorem{theorem}{theorem}

\newtheorem{lemma}[theorem]{lemma}
\newtheorem{proposition}[theorem]{proposition}

\newtheorem{definition}[theorem]{definition}

\newcommand{\abs}[1]{\left\vert#1\right\vert}
\newcommand{\set}[1]{\left\{#1\right\}}

\newcommand{\CommentS}[1]{}
{\vglue 5pt \VerbatimEnvironment\begin{Sbox}\begin{minipage}{0.9\textwidth}\begin{small}\begin{Verbatim}}%
{\end{Verbatim}\end{small}\end{minipage}\end{Sbox}\setlength{\shadowsize}{2pt}\shadowbox{\TheSbox}\vglue 5pt}

\usepackage{pgf}
\usepackage{tikz}
\usetikzlibrary{arrows}
\usetikzlibrary{shapes}
\tikzstyle{weight} = [font=\small]
\tikzstyle{edge} = [draw,thick,->,every node/.style={font=\sffamily\small}]
\tikzstyle{vertex}=[circle,fill=blue!20,minimum size=20pt,inner sep=0pt]
\usepackage{verbatim}
\usepackage{listings}
\usepackage[ruled,vlined]{algorithm2e}

\lstdefinelanguage{Smalltalk}{
  morekeywords={self,super,true,false,nil,thisContext}, % This is overkill
  morestring=[d]',
  morecomment=[s]{"}{"},
  alsoletter={\#:},
  escapechar={!},
  literate=
    {BANG}{!}1
    {UNDERSCORE}{\_}1
    {\\st}{Smalltalk}9 % convenience -- in case \st occurs in code
    {_}{{$\leftarrow$}}1
    {>>>}{{\sep}}1
    {^}{{$\uparrow$}}1
    {~}{{$\sim$}}1
    {-}{{\sf -\hspace{-0.13em}-}}1  % the goal is to make - the same width as +
    {-->}{{\quad$\longrightarrow$\quad}}3
	, % Don't forget the comma at the end!
  tabsize=2
}[keywords,comments,strings]

\title{A Simple FPTAS for Counting Edge Covers}
\date{}
\begin{document}
\author{
	Chengyu Lin
	\thanks{Shanghai Jiao Tong University. {\tt linmrain@gmail.com}}
	\and
	Jingcheng Liu
	\thanks{Shanghai Jiao Tong University. {\tt liuexp@gmail.com}}
	\and
	Pinyan Lu\thanks{Microsoft Research. {\tt pinyanl@microsoft.com}}
}
\maketitle
\begin{abstract}
An edge cover of a graph is a set of edges such that every vertex has at least an adjacent edge in it. We design a very simple deterministic fully polynomial-time approximation scheme  (FPTAS) for counting the number of edge covers for any graph. Previously, approximation algorithm is only known for 3 regular graphs and it is randomized. Our main technique is correlation decay, which is a powerful tool to design FPTAS for counting problems. In order to get FPTAS for general graphs without degree bound, we make use of a stronger notion called computationally efficient correlation decay, which is introduced in~\cite{LLY12}.
\end{abstract}

\section{Introduction}
An edge cover of a graph is a set of edges such that every vertex has at least one adjacent edge in it.
 For any graph without isolated vertices, there is at least one edge cover: the set of all edges. So the decision problem is trivial. There is also a polynomial time algorithm based on maximum matching to compute an edge cover with minimum cardinality.     In this paper, we study the counting version:
 For a given input graph, we count the number of edge covers for that graph. Unlike the decision or optimization problem, counting edge covers is a \#P-complete problem even when we restrict the input to 3 regular graphs. In this paper, we study the approximation version. For any given parameter $\epsilon>0$, the algorithm outputs a number $\hat{N}$ such that $(1-\epsilon) N\leq \hat{N} \leq (1+\epsilon) N$, where $N$ is the accurate number of edge covers of the input graph. We also require that the running time of the algorithm is bounded by $poly(n,1/\epsilon)$, where $n$ is the number of vertices of the given graph. This is called a fully polynomial-time approximation scheme (FPTAS). Our main result of this paper is an FPTAS for counting edge covers for any graph. Previously, approximation algorithm was only known for 3 regular graphs and the algorithm is randomized~\cite{MFCS09}. The randomized relaxation of FPTAS is called fully polynomial-time randomized approximation scheme (FPRAS), which uses random bits in the algorithm and requires that the final output is within the range $[(1-\epsilon) N, (1+\epsilon) N]$ with high probability.

Edge cover is related to many other graph problems such as (perfect) matching, $k$-factor problems and so on. All these problems involve a set of edges satisfying some local constraints defined on each vertex. For edge cover, it says that at least one incident edge should be chosen; while for matching, it is at most one edge. For generic constraints, it is the Holant framework~\cite{STOC09,holant}, which is well studied in terms of exact counting~\cite{holant,HuangL12,CaiGW13}, and recently in approximate counting~\cite{YZ13,McQuillan2013,fibo-approx}. For counting matchings, there is an FPRAS based on Markov Chain Monte Carlo (MCMC) for any graph~\cite{jerrum1989approximating}. Deterministic FPTAS is only known for graphs with bounded degree~\cite{BGKNT07}. For counting perfect matchings, it is a long standing open question if there is an FPRAS (or FPTAS) for it. For bipartite graphs, there is an FPRAS for counting perfect matchings. The weighted version can be viewed as  computing permanent of a non-negative matrix~\cite{app_JSV04}. This is one great achievement of approximate counting. It is still widely open if there exists an FPTAS for it or not. The current best deterministic algorithm can only approximate the permanent with an exponential large factor~\cite{linial1998deterministic,gamarnik2010deterministic}. There are many other counting problems, for which there is an FPRAS and we do not know if there is an FPTAS or not. In this paper, we give a complete FPTAS for a problem, for which even FPRAS was only known for very special family of instances.

Another view point of the edge cover problem is read twice monotone CNF formula (Rtw-Mon-CNF): Each edge is viewed as a Boolean variable and it is connected with two vertices (read twice); the constraint on each vertex is exactly a monotone OR constrain as at least one edge variable is assigned to be True. Counting number of solutions for a Boolean formula is another set of interesting problems studied both in exact counting and approximate counting. One famous example is the FPRAS for counting the solutions for a  DNF formula~\cite{KarpL83,KarpLM89}. It is an important open question to derandomize it~\cite{Trevisan04,gopalan2012dnf}. Our FPTAS for counting edge covers can also be viewed as an FPTAS for counting solutions for a  Rtw-Mon-CNF formula. If we do not restrict that each variable appears in at most two constraints, there is no FPTAS or FPRAS unless NP is equal to P or RP~\cite{dicho_DGJ10}.

The common overall approach for designing approximate counting algorithms is to relate counting with probability distribution.
In the context of randomized counting, this is usually referred as ``counting vs sampling" paradigm. If we can compute (or estimate) the marginal probability, which in our problem is the probability of a given edge is chosen when we sample an edge cover uniformly at random, we can in turn approximate the count. In randomized FPRAS, one estimates the marginal probability by sampling, and the most successful approach is sampling by Markov chain~\cite{MC_JA96}.
In deterministic FPTAS, one calculates the marginal probability directly, and the most successful approach is correlation decay as introduced in \cite{BG08} and \cite{Weitz06}. We elaborate a bit on the ideas.
 The marginal probability is estimated using only a local neighborhood around the edge. To justify the precision of the estimation, we show that far-away edges have little influence on the marginal probability.
One most successful example is in anti-ferromagnetic two-spin systems~\cite{LLY12,SST,LLY13}, including counting independent sets~\cite{Weitz06}. The correlation decay based FPTAS is beyond the best known MCMC based FPRAS and achieves the boundary of approximability~\cite{SS12,galanis2012inapproximability}.
To the best of our knowledge, that was the only example for which the best tractable range for correlation decay based FPTAS exceeds the sampling based FPRAS. This paper provides another such example. FPRAS was \emph{the} solution concept for approximate counting~\cite{dich_DGGJ00}. The recent development of correlation decay based FPTAS is changing the picture. It is interesting to investigate the deep relation between these two approaches.

A number of tools were developed for establishing correlation decay property: self-avoiding walk tree, computation tree, potential function, dangling instance, bounded variables and so on. These are something like coupling argument, canonical path and so on to establish the rapid mixing property of Markov Chains~\cite{MC_JA96}.
Armed with these powerful tools, there are recently many FPTAS's designed for various counting problems~\cite{LLY12,SST,LLY13,YZ13,LY13,fibo-approx}.
Many of these techniques are also used in this paper for designing and analyzing the FPTAS for counting edge covers.

Usually, the correlation decay property only implies FPTAS for system with bounded degree. The reason is that
we need to explore a local neighborhood with radius of order $\log n$, then the total running time is $n^{\log n}$ if there is no degree bound. To overcome this, we make use of a stronger notion called computationally efficient correlation decay
as introduced in~\cite{LLY12}. The observation is that when we go through a vertex with super-constant degree, the error is also decreased by a super-constant rate. Thus we do not need to explore a depth of $\log n$ if the degrees are large. The tradeoff between degree and decay rate defined by computationally efficient correlation decay can support FPTAS with unbounded degree systems. Previously, this notion was only used in anti-ferromagnetic two-spin systems. In this paper, we prove that the distribution defined by edge covers also satisfies this stronger version of correlation decay and thus we give FPTAS for counting edge covers for any graph.

\section{Preliminaries}
%\subsection{Definitions}
An edge cover of a graph is a set of edges such that every vertex has at least one adjacent edge in it.
Given a graph $G=(V,E)$ with $e \in E$,  we use $EC(G)$ to denote the set of all edge covers of graph $G$, and $P(G, e)$ to denote the marginal probability over $EC(G)$ that edge $e$ is \emph{not} chosen, or formally, with $X \sim EC(G)$ uniformly,
\begin{equation}
	P(G, e) \triangleq \mathbb{P} \left(\textrm{edge $e$ is not chosen in $X$} \right)
	\label{defpge}
\end{equation}

In this paper, we deal with an extended notion of undirected graphs where dangling edges and free edges are allowed.
\begin{definition}
	A {\bf dangling edge} $e=(u,\_)$ of a graph is such singleton edge with exactly one end-point vertex $u$, as shown in the Figure \ref{fig:G}.

	A {\bf free edge} $e=(\_, \_)$ of a graph is such edge with no end-point vertex. %Note that a free edge is not a dangling edge.

\end{definition}

	We use graph to refer graph with or without dangling edges and free edges.
	Edges in the usual sense (i.e. neither dangling nor free), will be referred to as normal edges.
	%and graphs in the usual sense(i.e. graphs with only normal edges) will be refered to as normal graphs.

	We remark that an alternative view to these combinatorial definitions is from Rtw-Mon-CNF.
	A dangling edge is simply a variable which only appears at one clause, and a free edge is a variable
	that does not appear at all, whereas normal edge just corresponds to variables appearing twice.

For a graph $G=(V,E)$, an edge $e = (u,v) \in E$ and a vertex $u \in V$, define
\begin{align*}
    G - e \triangleq& (V, E-e) \\
    e - u \triangleq& (\_, v) \text{(note that here $v$ could be $\_$)} \\
    G - u \triangleq& (V - u, \\
                    &\set{e: e \in E, e\text{ is not incident with }u} \\
                    &\cup \set{e - u: e \in E, e\text{ is incident with }u})
\end{align*}

%For a graph $G=(V,E)$ with an edge (may be dangling or free) $e \in E$ and a vertex $v \in V$,
%$$G-e \triangleq (V, E-e)$$
%and
%\begin{align*}
%G-v \triangleq & (V-\set{v}, \set{(x,y):x \in V-\set{v}, y \in V-\set{v}, (x,y) \in E}  \\
% &\cup \set{(x,\_):x \in V-\set{v}, (v,x) \in E}  \\
% &\cup \set{(\_,\_): (v,\_) \in E \text{ or } (v,v) \in E})
%\end{align*}

Note that here in edge set $E$, duplicates are allowed. We may have multiple dangling edges $(v,\_)$ and many free edges $(\_,\_)$. Recall that here edges are unordered pairs so we treat $(v,\_)$ and $(\_,v)$ as the same.

For example, given a degree-3 vertex $u$ with dangling edge $e$ shown in Figure \ref{fig:G} , the result of $e_1 - u$ is shown in Figure \ref{fig:e-u} and the result of $G-e-u\triangleq (G-e)-u$ is shown in Figure \ref{fig:G-e-u}.

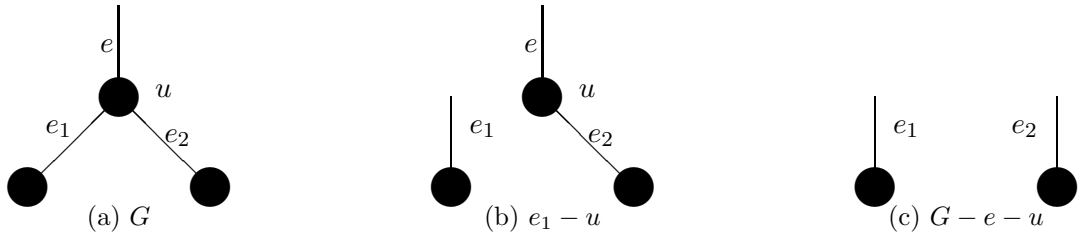
\begin{figure}[htp]
	\begin{subfigure}[b]{0.30\textwidth}
		\centering
		\setlength{\unitlength}{1.2mm}
		\begin{picture}(20,20)
			\put(0,0){\circle*{6}}
			\put(0,0){\line(1,1){10}}
			\put(20,0){\circle*{6}}
			\put(20,0){\line(-1,1){10}}
			\put(10,10){\circle*{6}}
			\put(10,10){\line(0,1){10}}
			\put(14,10){$u$}
			\put(8,15){$e$}
			\put(2,6){$e_1$}
			\put(15,5){$e_2$}
		\end{picture}
		\caption{$G$}
		\label{fig:G}
	\end{subfigure}
	\hfill
    \begin{subfigure}[b]{0.30\textwidth}
		\centering
		\setlength{\unitlength}{1.2mm}
		\begin{picture}(20,20)
			\put(0,0){\circle*{6}}
			\put(0,0){\line(0,1){10}}
			\put(20,0){\circle*{6}}
			\put(20,0){\line(-1,1){10}}
			\put(10,10){\circle*{6}}
			\put(10,10){\line(0,1){10}}
			\put(14,10){$u$}
			\put(8,15){$e$}
			\put(2,6){$e_1$}
			\put(15,5){$e_2$}
		\end{picture}
		\caption{$e_1-u$}
		\label{fig:e-u}
	\end{subfigure}
    \hfill
	\begin{subfigure}[b]{0.30\textwidth}
		\centering
		\setlength{\unitlength}{1.2mm}
		\begin{picture}(20,20)
			\put(0,0){\circle*{6}}
			\put(0,0){\line(0,1){10}}
			\put(20,0){\circle*{6}}
			\put(20,0){\line(0,1){10}}
			\put(2,6){$e_1$}
			\put(15,6){$e_2$}
		\end{picture}
		\caption{$G-e-u$}
		\label{fig:G-e-u}
	\end{subfigure}
	\caption{Dangling edges examples.}
\end{figure}

We use $0$ to denote scalar value $0$, and $\mathbf{0}$ to denote the all-zero vector, and $\set{e_i}_{i=1}^{d}$ denote the $d$-dimensional vector with $i$-th coordinate being $e_i$, so $\set{e_i} = \mathbf{0}$ means $\forall i, e_i = 0$.
We also use the convention that when $d=0, \prod_i^d p_i \triangleq 1$.

In general we use $n$ to refer to the number of vertices in a given graph, and $m$ to refer to the number of edges.

\section{The Computation Tree Recursion}
In this section, we provide a recursion for computing the marginal probability $P(G, e)$ with that of smaller instances.

\subsection{$e$ is free}
%This case is trivial.
\begin{proposition}
	\[P(G,e) = \frac{1}{2}\]
\end{proposition}
\begin{proof}
	If $e$ is a free edge, then any edge cover with $e$ chosen is in one-to-one correspondence to an edge cover with $e$ not chosen. Hence exactly half of the edge covers in $EC(G)$ does not choose $e$, so $P(G, e) = \frac{1}{2}$.
\end{proof}

\subsection{$e$ is dangling}
\begin{lemma}
For graph $G=(V,E)$ with a dangling edge $e=(u,\_)$, denote $d$
edges incident with $u$ except $e$ as $e_1, e_2, \ldots, e_d$,
let $G_1 \triangleq G - e - u$, and $\forall i \geq 2, G_i \triangleq G_{i-1} - e_{i-1}$ ,
	\begin{equation}
		P(G, e) = \frac{1-\prod_{i=1}^d P(G_i, e_i)}{2 - \prod_{i=1}^d P(G_i, e_i)} %= \frac{1}{2} - \frac{0.5 \prod_{i=1}^d P(G_i, e_i)}{2 - \prod_{i=1}^d P(G_i, e_i)}
		\label{propp3rg}
	\end{equation}
\end{lemma}
\begin{proof}
	For $\boldsymbol\alpha \in \set{0,1}^d$, let $EC_{\boldsymbol\alpha}(G-e-u)$ be the set of edge covers in $G-e-u$ such that its restriction onto $\set{e_i}_{i=1}^{d}$ is consistent with $\boldsymbol\alpha$, denote $Z_{\boldsymbol\alpha} = \abs{EC_{\boldsymbol\alpha}(G-e-u)}$, and $Z = \sum_{\boldsymbol\alpha \in \set{0,1}^d} Z_{\boldsymbol\alpha} = \abs{EC(G-e-u)}$. % \triangleq \set{X : X\subseteq E, $

		Also note that as long as $\boldsymbol\alpha \neq 0$, counting edge covers with restriction $\boldsymbol\alpha$ is the same in either $G$, $G-e$, or $G-e-u$, so it is enough to work with $G-e-u$. Note that in $G-e-u$, for every $i$, $e_i$ is either dangling or free, but not normal.
	\begin{align*}
		P(G,e) = & \frac{\abs{EC(G-e)}}{\abs{EC(G)}} \\
        = & \frac{\sum_{\boldsymbol\alpha \in \set{0,1}^d, \boldsymbol\alpha \neq \mathbf{0}} Z_{\boldsymbol\alpha} }{ Z_{\mathbf{0}} + 2 \sum_{\boldsymbol\alpha \in \set{0,1}^d, \boldsymbol\alpha \neq \mathbf{0}} Z_{\boldsymbol\alpha}} \\
        = &\frac{Z - Z_0}{2 Z - Z_0} \\
        = &\frac{1 - \frac{Z_{\mathbf{0}}}{Z}}{ 2 - \frac{Z_{\mathbf{0}}}{Z}}.
	\end{align*}

	Now consider the term $\frac{Z_{\mathbf{0}}}{Z}$, it says the probability that a uniformly random edge cover drawn from $EC(G-e-u)$ picked none of $\set{e_i}_{i=1}^{d}$, so
	\begin{align*}
        \frac{Z_{\mathbf{0}}}{Z} = & \mathbb{P} \left( \set{e_i} = \mathbf{0}\right) \\
        = & \mathbb{P} (e_1 = 0) \prod_{i=2}^d \mathbb{P} \left(e_i = 0 \mid \set{e_j}_{j=1}^{i-1} =\mathbf{0}\right) \\
        = &\prod_{i=1}^d P(G_i, e_i)
	\end{align*}

	Hence by substitution we have
    $$P(G, e) = \frac{1-\prod_{i=1}^d P(G_i, e_i)}{2 - \prod_{i=1}^d P(G_i, e_i)}$$
	
\end{proof}

We remark that for every $i$, $e_i$ is dangling or free in $G_i$.

\subsection{$e$ is a normal edge}

	For $e=(u,v)$ as a normal edge, let $\set{e_i}$ be the set of edges incident with vertex $u$ except $e$, and $\set{f_i}$ be the set of edges incident with vertex $v$ except $e$, and $d_1 = \abs{\set{e_i}}, d_2 = \abs{\set{f_i}}$, now for $\boldsymbol\alpha \in \set{0,1}^{d_1}, \boldsymbol\beta \in \set{0,1}^{d_2}$, we use $EC_{\boldsymbol\alpha,\boldsymbol\beta}(G)$ to denote the set of edge covers for $G$ such that its restriction to $\set{e_i}_{i=1}^{d_1}$ is consistent with $\boldsymbol\alpha$, and restriction to $\set{f_i}_{i=1}^{d_2}$ is consistent with $\boldsymbol\beta$.

	Denote $Z_{\boldsymbol\alpha, \boldsymbol\beta}^G \triangleq \abs{EC_{\boldsymbol\alpha, \boldsymbol\beta}(G)}$, $G' \triangleq G-e, G'' \triangleq G-e-u-v$. As an illustration, given a normal edge $e=(u,v)$ in $G$ as in Figure \ref{fig:generalG}, $G'$ and $G''$ are Figure \ref{fig:generalG-e} and Figure \ref{fig:generalG-e-u-v} respectively.

\begin{figure}[htp]
	\begin{subfigure}[b]{0.30\textwidth}
		\centering
		\setlength{\unitlength}{1.2mm}
		\begin{picture}(20,30)
			\put(0,0){\circle*{6}}
			\put(0,0){\line(1,1){10}}
			\put(20,0){\circle*{6}}
			\put(20,0){\line(-1,1){10}}
			\put(10,10){\circle*{6}}
			\put(10,10){\line(0,1){10}}
			\put(10,20){\circle*{6}}
			\put(10,20){\line(1,1){10}}
			\put(10,20){\line(-1,1){10}}
			\put(20,30){\circle*{6}}
			\put(0,30){\circle*{6}}
			\put(14,20){$v$}
			\put(14,10){$u$}
			\put(8,14.5){$e$}
			\put(2,6){$e_1$}
			\put(15,5){$e_2$}
			\put(3,28){$f_1$}
			\put(13,28){$f_2$}
		\end{picture}
		\caption{$G$}
		\label{fig:generalG}
	\end{subfigure}
	\hfill
	\begin{subfigure}[b]{0.30\textwidth}
		\centering
		\setlength{\unitlength}{1.2mm}
		\begin{picture}(20,30)
			\put(0,1){\circle*{6}}
			\put(0,0){\line(1,1){10}}
			\put(20,1){\circle*{6}}
			\put(20,0){\line(-1,1){10}}
			\put(10,10){\circle*{6}}
			\put(10,20){\circle*{6}}
			\put(10,20){\line(1,1){10}}
			\put(10,20){\line(-1,1){10}}
			\put(20,30){\circle*{6}}
			\put(0,30){\circle*{6}}
			\put(14,20){$v$}
			\put(14,10){$u$}
			\put(2,6){$e_1$}
			\put(15,5){$e_2$}
			\put(3,28){$f_1$}
			\put(13,28){$f_2$}
		\end{picture}
		\caption{$G'$}
		\label{fig:generalG-e}
	\end{subfigure}
	\hfill
	\begin{subfigure}[b]{0.30\textwidth}
		\centering
		\setlength{\unitlength}{1.2mm}
		\begin{picture}(20,20)
			\put(0,1){\circle*{6}}
			\put(0,1){\line(0,1){10}}
			\put(20,1){\circle*{6}}
			\put(20,1){\line(0,1){10}}
			\put(20,30){\circle*{6}}
			\put(20,30){\line(0,-1){10}}
			\put(0,30){\circle*{6}}
			\put(0,30){\line(0,-1){10}}
			\put(2,6){$e_1$}
			\put(15,6){$e_2$}
			\put(2,24){$f_1$}
			\put(15,24){$f_2$}
		\end{picture}
		\caption{$G''$}
		\label{fig:generalG-e-u-v}
	\end{subfigure}
	\caption{Normal edge examples.}
\end{figure}
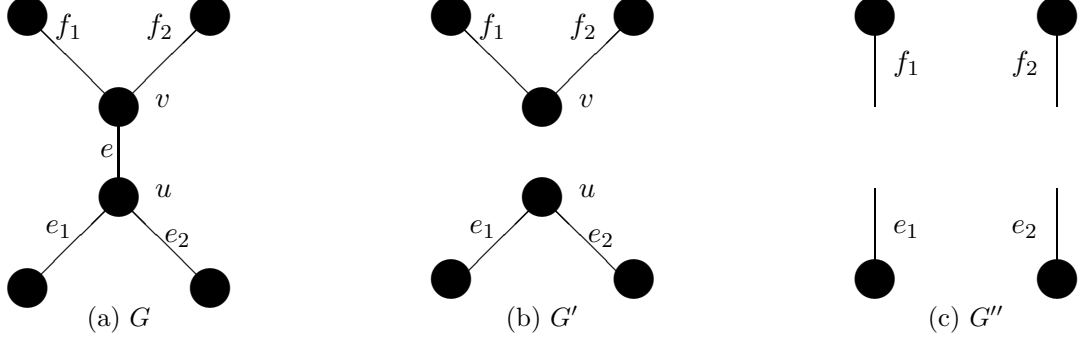

By definition we have

\begin{equation}
	P(G,e) = \frac{\abs{EC(G')}}{\abs{EC(G')} + \abs{EC(G'')} }
\end{equation}

	Note that as long as $\boldsymbol\alpha \neq \mathbf{0} , \boldsymbol\beta \neq \mathbf{0}$, working with $G'$ and working with $G''$ is the same with restriction to $\boldsymbol\alpha$ and $\boldsymbol\beta$, or formally,
\[\abs{EC(G')} = \sum_{\boldsymbol\alpha \neq \mathbf{0}, \boldsymbol\beta \neq \mathbf{0}} Z_{\boldsymbol\alpha, \boldsymbol\beta}^{G'} = \sum_{\boldsymbol\alpha \neq \mathbf{0}, \boldsymbol\beta \neq \mathbf{0}} Z_{\boldsymbol\alpha, \boldsymbol\beta}^{G''}\]
%\[C_2 = \sum_{\boldsymbol\alpha , \boldsymbol\beta} Z_{\boldsymbol\alpha, \boldsymbol\beta}^{G''}\]

Since only $G''$ is involved, denote

$Z_{\boldsymbol\alpha, \boldsymbol\beta} \triangleq Z_{\boldsymbol\alpha, \boldsymbol\beta}^{G''}$

$Z \triangleq \sum_{\boldsymbol\alpha \in \set{0,1}^{d_1} , \boldsymbol\beta \in \set{0,1}^{d_2}} Z_{\boldsymbol\alpha, \boldsymbol\beta}$

    $G_i^1 \triangleq G'' - \sum_{k=1}^{i-1} e_k$

    $G_i^2 \triangleq G'' - \sum_{k=1}^{d_1}e_k - \sum_{k=1}^{i-1} f_k$

    $G_i^3 \triangleq G'' - \sum_{k=1}^{i-1} f_k$

\begin{lemma}
	Let $X = \prod_{i=1}^{d_1} P(G_i^1, e_i)$, 

$Y = \prod_{i=1}^{d_2} P(G_i^2, f_i), Z = \prod_{i=1}^{d_2} P(G_i^3, f_i)$,
	\begin{equation}
	P(G,e) =  1 - \frac{1}{2 + X \cdot Y - X - Z}
		\label{recursion2}
	\end{equation}

\end{lemma}
\begin{proof}

	\begin{align*}
P(G,e) &= \frac{\sum_{\boldsymbol\alpha \neq \mathbf{0}, \boldsymbol\beta \neq \mathbf{0}} Z_{\boldsymbol\alpha, \boldsymbol\beta}}{Z + \sum_{\boldsymbol\alpha \neq \mathbf{0}, \boldsymbol\beta \neq \mathbf{0}} Z_{\boldsymbol\alpha, \boldsymbol\beta}} \\
&=\frac{Z - \sum_{\boldsymbol\alpha}Z_{\boldsymbol\alpha,\mathbf{0}} - \sum_{\boldsymbol\beta} Z_{\mathbf{0}, \boldsymbol\beta} + Z_{\mathbf{0}, \mathbf{0}}}{2Z - \sum_{\boldsymbol\alpha}Z_{\boldsymbol\alpha,\mathbf{0}} - \sum_{\boldsymbol\beta} Z_{\mathbf{0}, \boldsymbol\beta} + Z_{\mathbf{0}, \mathbf{0}}} \\
&= 1 - \frac{1}{2 + \frac{Z_{\mathbf{0},\mathbf{0}}}{Z} - \frac{\sum_{\boldsymbol\beta} Z_{\mathbf{0}, \boldsymbol\beta} }{Z} - \frac{\sum_{\boldsymbol\alpha} Z_{ \boldsymbol\alpha , \mathbf{0}} }{Z}}
	\end{align*}

	Denote $\mathbb{P} \left( \boldsymbol\alpha = \mathbf{0}, \boldsymbol\beta = \mathbf{0} \right) \triangleq \frac{Z_{\mathbf{0},\mathbf{0}}}{Z}$,
 
 $\mathbb{P} \left( \boldsymbol\alpha = \mathbf{0} \right) \triangleq \frac{\sum_{\boldsymbol\beta} Z_{\mathbf{0}, \boldsymbol\beta} }{Z}, \mathbb{P} \left( \boldsymbol\beta = \mathbf{0} \right) \triangleq \frac{\sum_{\boldsymbol\alpha} Z_{ \boldsymbol\alpha , \mathbf{0}} }{Z}$.

Now consider the three terms respectively,
	\begin{align*}
        \mathbb{P}\left( \boldsymbol\alpha = \mathbf{0}\right) =& \mathbb{P} \left( \set{e_i} = \mathbf{0}\right) =	\prod_{i=1}^{d_1} P(G_i^1, e_i) \\
        \mathbb{P}\left( \boldsymbol\beta = \mathbf{0}\right) =& \mathbb{P} \left( \set{f_i} = \mathbf{0}\right) =	\prod_{i=1}^{d_2} P(G_i^3, f_i)
    \end{align*}
	\begin{align*}
        &\mathbb{P}\left( \boldsymbol\alpha = \mathbf{0}, \boldsymbol\beta = \mathbf{0} \right) \\
        =&  \mathbb{P} \left( \boldsymbol\alpha = \mathbf{0} \right) \cdot \mathbb{P}\left( \boldsymbol\beta = \mathbf{0} \mid \boldsymbol\alpha = \mathbf{0} \right) \\
        =&  \mathbb{P} \left( \set{e_i} = \mathbf{0}\right) \cdot \mathbb{P} \left( \set{f_i} = \mathbf{0} \mid \set{e_i} = \mathbf{0} \right) \\
        =& \prod_{i=1}^{d_1} \mathbb{P} \left( e_i = 0 \mid \set{e_j}_{j=1}^{i-1} = \mathbf{0} \right) \\
         &\cdot \prod_{i=1}^{d_2} \mathbb{P} \left( f_i = 0 \mid \set{e_j}_{j=1}^{d_1} = \mathbf{0},\set{f_j}_{j=1}^{i-1} = \mathbf{0} \right) \\
        =& \prod_{i=1}^{d_1} P(G_i^1, e_i) \cdot \prod_{i=1}^{d_2} P(G_i^2, f_i)
	\end{align*}

	Hence equation (\ref{recursion2}) is verified.
\end{proof}

Remark that for every $i$, $e_i$ is dangling or free in $G_i^1$, $f_i$ is dangling or free in $G_i^3$, and in $G_i^2$, neither $e_i$ nor $f_i$ is normal.

\section{Estimating Marginal Probability}

%\IncMargin{1em}
%\begin{algorithm}[H]
\begin{algorithm}
\SetKwInOut{Input}{input}\SetKwInOut{Output}{output}
\emph{ \textbf{function} $P(G, e, L):$}
\BlankLine
\Input{Graph $G$; edge $e$; Recursion depth $L$; }
\Output{Estimate of $P(G,e)$ up to depth $L$ .}
\Begin{
	\If{$L\leq0$ }{\Return{ $\frac{1}{2}$}}
	\ElseIf{$e$ is free }{
		\Return{ $\frac{1}{2}$}\;
	}
	\ElseIf{$e$ is dangling }{
		$L' \leftarrow L - \lceil \log_6{(d+1)} \rceil$\;
		%$L' \leftarrow L - \lceil d/5 \rceil$\;
		\Return{ $\frac{1-\prod_{i=1}^d P(G_i, e_i, L')}{2 - \prod_{i=1}^d P(G_i, e_i, L')}$} \;
	}
	\Else(// $e$ is normal ){
		$X \leftarrow \prod_{i=1}^{d_1} P(G_i^1, e_i, L)$\;
		$Y \leftarrow \prod_{i=1}^{d_2} P(G_i^2, f_i, L)$\;
		$Z \leftarrow \prod_{i=1}^{d_2} P(G_i^3, f_i, L)$\;
		\Return{ $1 - \cfrac{1}{2+ X\cdot Y  - X - Z }$ }\;
	}
 }
 \caption{Estimate $P(G,e)$}% up to depth $L$}
\end{algorithm}
%\DecMargin{1em}

We may compute the marginal probability $P(G, e)$ exactly with the previous recursion, but
that could take recursion depth of $O(n)$ which results in exponential computation time.
So here we use a truncated computation tree for an estimate of $P(G,e)$.

As a remark, the recursion depth used here is actually the so-called $M$-based depth introduced in \cite{LLY12} with $M=6$.
%In fact $M$ could take any value as long as $M \geq 6$. %, and a slightly extended recursion depth could improve the algorithm running time for large $d$.
%Note the recursion depth used here is just a natural generalization of the so-called $M$-based depth introduced in \cite{LLY12}, we remark that it's sufficient to get an FPTAS with $M\geq 6$ using the $M$-based depth, here we show a stronger efficient algorithm via a slightly modified notion of recursion depth.

%\subsection{Analysis of the Algorithm}

Note that the normal case is invoked only once, so the algorithm keeps exploring in the third cases, until it hits the first 2 cases. We remark that an alternative view of the recursion depth is, we replace every node with degree greater than 6 with a $6$-ary branching subtree.
Now with this alternative view, it is easy to see that the nodes involved in the branching tree up to depth $L$ is at most $6^L$,
and for the initial normal edge case it involves at most $n$ subtrees, and for second-to-base-case nodes (i.e. nodes with $0<L \leq \lceil \log_6{(d+1)} \rceil $ ) they involve at most $n$ extra base cases,
so the algorithm $P(G,e,L)$ has running time $O(n^2 \cdot 6^L)$.
%Now let $B(L)$ be the set of vertices in the recursion tree involved, and $R(L) \triangleq \norm{B(L)}$,
% by the recursion on $P(G,e,L)$ in the third case we have the recursive relation for $R(L)$,
%

 % Note that the corner case when $e$ is free so this is \leq rather than =
% \begin{align*}
%	 %R(L) \leq d R(L-d/5) , L > 0\\
%	 R(L) \leq d R(L-\lceil \log_6{(d+1)} \rceil) , L > 0\\
%	 R(L) = 1, L\leq 0
% \end{align*}

 %Therefore we conclude that $R(L) \leq d^{1+\frac{L}{\log_6{(d+1)}}} \leq d\cdot 6^L$, in other words, the running time of the above algorithm with recursion depth $L$ is at most $O(n\cdot 6^L)$.

\section{Correlation Decay Property}

In the last section, we show an algorithm $P(G,e,L)$ for estimating the marginal probability $P(G,e)$,
so here we establish the exponential correlation decay property, in the stronger sense with the
%modified recursion depth,
$M$-based depth,
of the estimation error in $P(G,e,L)$.%, hence $P(G,e,L)$

\begin{theorem}
	\label{cd-main-theorem}
	Given graph $G$, edge $e$ and depth $L$,
	\[\abs{P(G,e,L) - P(G,e)} \leq 3\cdot(\frac{1}{2})^{L+1}\]
\end{theorem}

Such phenomenon is usually referred to as exponential correlation decay. Before we prove the main theorem, we will introduce a few useful propositions and lemmas.

\begin{proposition}
	\[P(G, e) \leq \frac{1}{2}\]
\end{proposition}

\begin{proof}
	Although one may examine this case by case algebraically, this proposition is quite obvious in a combinatorial view, for any edge cover $X \in EC(G)$ with $e \notin X$, $X+e$ is also an edge cover in $G$, and $\forall X,Y \in EC(G)$ s.t. $X \neq Y, e \notin X, e\notin Y$, we have $X+e \neq Y+e$. So the edge covers with $e$ chosen is at least as many as the edge covers with $e$ not chosen, hence the proposition follows.
\end{proof}

We remark that our algorithm also guarantees that $P(G,e,L) \leq \frac{1}{2}$,
since for the dangling case, $\frac{1 - \prod_i x_i}{2 - \prod_i x_i} = \frac{1}{2} - \frac{\prod_i x_i}{2(2 - \prod_i x_i)}$ ; and for normal case $X\cdot Y - X - Z \leq 0$.

For notational convenience, given $d$-dimensional vector ${\bf x} \in [0, \frac{1}{2}]^d$, we denote
\[ f({\bf x}) \triangleq \frac{1- \prod_i x_i}{2 - \prod_i x_i}\]

Given a $d_1$-dimensional vector ${\bf x} \in [0, \frac{1}{2}]^{d_1}$ and two $d_2$-dimensional vectors ${\bf y,z} \in [0, \frac{1}{2}]^{d_2}$, let
\[ g({\bf x,y,z}) \triangleq  1- \frac{1}{2+\prod_i x_i \cdot \prod_i y_i - \prod_i x_i - \prod_i z_i} \]

	\begin{lemma}
		\label{meanvalue1}
		For $d$-variate function $f$, given estimated $\hat{\bf x}$ for true value ${\bf x}$ such that $\hat{\bf x} \in  [0, \frac{1}{2}]^d, {\bf x} \in [0, \frac{1}{2}]^d$, let $\epsilon \triangleq \max_i{\abs{x_i - \hat{x_i}}}$,
		\begin{align*}
			\abs{f(\hat{\bf x}) - f({\bf x})}
		\leq  \min{\set{\frac{1}{2}, d \left( \frac{1}{2} \right)^{d-1}}} \cdot \epsilon
		\end{align*}
	\end{lemma}

	\begin{proof}
        First for $d$-dimensional vector ${\bf x} \in [0, \frac{1}{2}]^d$,
        \begin{align}
			\label{mv1pd}
			\sum_i^d \abs{\cfrac{\partial f({\bf x})}{\partial x_i}} \leq & \min{\set{\frac{1}{2}, d \left( \frac{1}{2} \right)^{d-1}}}
		\end{align}

    	For $d=0, \sum_i\abs{ \cfrac{\partial f({\bf x})}{\partial x_i}} = 0$.

    	For $d=1, \sum_i\abs{ \cfrac{\partial f({\bf x})}{\partial x_i}} = \frac{1}{\left( 2 - x_1 \right)^2} \leq \frac{4}{9} $.

    	For $d=2, \sum_i\abs{ \cfrac{\partial f({\bf x})}{\partial x_i}} = \frac{x_1 + x_2}{\left( 2 - x_1x_2 \right)^2} \leq \frac{16}{49} $.

    	For $d=3, \sum_i\abs{ \cfrac{\partial f({\bf x})}{\partial x_i}} = \frac{x_1 x_2 + x_1 x_3 + x_2 x_3}{\left( 2 - x_1x_2x_3 \right)^2} \leq \frac{16}{75} $.

	Next by $\forall k, x_k \leq \frac{1}{2}$, 
    	\begin{align*}
    		\sum_i^d \abs{\cfrac{\partial f({\bf x})}{\partial x_i}}  =& \sum_i^d \cfrac{ \prod_{k \neq i}^d x_k  }{\left( 2 - \prod_i x_i \right)^2}
    		%\leq  d \prod_{k \neq i^*}^d x_k
    		\leq  d \left( \frac{1}{2} \right)^{d-1}
    	\end{align*}

		So for $d \geq 4, \sum_i \abs{\cfrac{\partial f({\bf x})}{\partial x_i}} \leq \frac{1}{2}$, we have verified the inequality relation (\ref{mv1pd}).

        Now let $h(\alpha) = f(\alpha {\bf x} + (1 - \alpha) {\bf \hat{x}})$ where $\alpha \in [0,1]$, by (\ref{mv1pd}) and Mean Value Theorem, $\exists \tilde{\alpha} \in [0,1]$ s.t. for $\tilde{\bf x} = \tilde{\alpha} {\bf x} + (1 - \tilde{\alpha}) {\bf \hat{x}}$

		\begin{align*}
			\abs{f(\hat{\bf x}) - f({\bf x})} \leq& \sum_i\abs{ \frac{\partial f(\tilde{\bf x})}{\partial x_i}} \cdot \epsilon \\
            \leq&  \min{\set{\frac{1}{2}, d \left( \frac{1}{2} \right)^{d-1}}} \cdot \epsilon
		\end{align*}
	\end{proof}

	\begin{lemma}
		\label{meanvalue2}
		Given estimated $\hat{\bf x},\hat{\bf y},\hat{\bf z}$ for true value $ {\bf x}, {\bf y}, {\bf z}$ respectively, such that ${\bf x}, \hat{\bf x} \in [0, \frac{1}{2}]^{d_1} , {\bf y,z}, \hat{\bf y},\hat{\bf z} \in [0, \frac{1}{2}]^{d_2}$,
		let $\epsilon \triangleq \max_i \set{\abs{x_i - \hat{x_i}}, \abs{y_i - \hat{y_i} } , \abs{z_i - \hat{z_i}}}$,
		\begin{align*}
		\abs{g(\hat{\bf x}, \hat{\bf y}, \hat{\bf z}) - g({\bf x,y,z})}
		\leq  3\epsilon
		\end{align*}

	\end{lemma}

	\begin{proof}
Denote

$S_1({\bf x,y,z}) \triangleq \sum_k^{d_1} \abs{ \frac{\partial g({\bf x,y,z})}{\partial x_k} }$,

$S_2({\bf x,y,z}) \triangleq \sum_k^{d_2} \abs{ \frac{\partial g({\bf x,y,z})}{\partial y_k} }$,

$S_3({\bf x,y,z}) \triangleq \sum_k^{d_2} \abs{ \frac{\partial g({\bf x,y,z})}{\partial z_k} }$,

$S_4({\bf x,y,z}) \triangleq S_1({\bf x,y,z})+S_2({\bf x,y,z})+S_3({\bf x,y,z})$,

$X \triangleq \prod_i^{d_1} x_i$,
$Y \triangleq \prod_i^{d_2} y_i$,
$Z \triangleq \prod_i^{d_2} z_i$,

$W \triangleq (2+X Y-X-Z)^2$.

For ${\bf x} \in [0, \frac{1}{2}]^{d_1} , {\bf y,z} \in [0, \frac{1}{2}]^{d_2}$,
		\begin{align*}
		S_1({\bf x,y,z}) &= \frac{1}{W} \sum_k^{d_1} \prod_{i\neq k}^{d_1} x_i \cdot \left( 1 - Y \right)
		\leq \frac{d_1}{2^{d_1 - 1}}  \leq 1 \\
		S_2({\bf x,y,z}) &= \frac{1}{W} \sum_k^{d_2} \prod_{i\neq k}^{d_2} y_i \cdot X
		\leq  \frac{d_2}{2^{d_1 + d_2 - 1}} \leq 1 \\
		S_3({\bf x,y,z}) &= \frac{1}{W} \sum_k^{d_2} \prod_{i\neq k}^{d_2} z_i
		\leq \frac{d_2}{2^{d_2 - 1}} \leq 1
		\end{align*}

        Now let $h(\alpha) = g(\alpha {\bf x} + (1 - \alpha) {\bf \hat{x}}, \alpha {\bf y} + (1 - \alpha) {\bf \hat{y}}, \alpha {\bf z} + (1 - \alpha) {\bf \hat{z}})$ where $\alpha \in [0,1]$.

         By Mean Value Theorem, $\exists \tilde{\alpha} \in [0,1]$ s.t. for $\tilde{\bf x} = \tilde{\alpha} {\bf x} + (1 - \tilde{\alpha}) {\bf \hat{x}}, \tilde{\bf y} = \tilde{\alpha} {\bf y} + (1 - \tilde{\alpha}) {\bf \hat{y}}, \tilde{\bf z} = \tilde{\alpha} {\bf z} + (1 - \tilde{\alpha}) {\bf \hat{z}}$
        %Then by Mean Value Theorem, $\exists \tilde{\bf x}, \tilde{\bf y},  \tilde{\bf z}$ s.t. $\forall i, \min\set{x_i,\hat{x_i}} \leq \tilde{x_i} \leq \max\set{x_i, \hat{x_i}}, \min\set{y_i,\hat{y_i}} \leq \tilde{y_i} \leq \max\set{y_i, \hat{y_i}}, \min\set{z_i,\hat{z_i}} \leq \tilde{z_i} \leq \max\set{z_i, \hat{z_i}}$.
		\begin{align*}
		\abs{g(\hat{\bf x}, \hat{\bf y}, \hat{\bf z}) - g({\bf x,y,z})} \leq &
		S_4(\tilde{\bf x},\tilde{\bf y},\tilde{\bf z} ) \epsilon
		\leq  3\epsilon
		\end{align*}
	\end{proof}

	%Now we are ready for the proof of  Theorem \ref{cd-main-theorem}.

	%\begin{proof}

\noindent{\bf Proof of Theorem \ref{cd-main-theorem}. }
		Note that the recursion for normal edge case is applied only once, so it is sufficient to show that for free or dangling edge $e$:

		\[\abs{P(G,e,L) - P(G,e)} \leq (\frac{1}{2})^{L+1}\]
		
		And it automatically follows from Lemma \ref{meanvalue2} that for normal edge $e$:

		\[\abs{P(G,e,L) - P(G,e)} \leq 3\cdot(\frac{1}{2})^{L+1}\]

		Now we prove by induction with induction hypothesis that for free or dangling edge $e$:

		\[\abs{P(G,e,L) - P(G,e)} \leq (\frac{1}{2})^{L+1}\]
		
		For base case $L=0, \abs{P(G,e,L) - P(G,e)} \leq \frac{1}{2}$ holds when $e$ is free or dangling.

		Now suppose the induction hypothesis is true for $L<k$, we shall prove that it is true for $L=k$.

		{\bf Case 1}, $e$ is free edge $\abs{P(G,e,L) - P(G,e)} = 0$.

		{\bf Case 2}, $e=(u,\_)$ is a dangling edge, denote with $deg(u)=d+1$, then by induction hypothesis we have
        $\epsilon \triangleq \max_i \abs{P\left(G_i,e_i,L-\lceil \log_6{(d+1)}\rceil \right) - P(G_i,e_i)} \leq \left(\frac{1}{2}\right)^{L-\lceil \log_6{(d+1)}\rceil + 1}$.

		First by Lemma \ref{meanvalue1} we need to show that for $d \leq 4$,
        \[\frac{1}{2^{1+L-\lceil \log_6{(d+1)}\rceil + 1}} \leq \frac{1}{2^{L+1}}\]

		which is obvious because $\lceil\log_6{(d+1)}\rceil \leq 1$.

		Next we show for $d \geq 5$,
        \[ d\cdot \left( \frac{1}{2} \right)^{d-1 + L - \lceil \log_6{(d+1)}\rceil + 1}  \leq \left( \frac{1}{2} \right)^{L + 1} \]

		Namely for $d \geq 5$,
		\[ \log_2 d + \lceil \log_6{(d+1)} \rceil \leq d-1\]

		For $d=5,6$, one can directly examine that as $\log_2 d < 3$ and $\log_6 6 =1, \log_6 7 < 2$.

        %For $d\geq7$, by simply taking the derivative one can show that
		%\[ \log_2 d + \log_6{(d+1)} \leq d-1\]
		% @mrain:The derivative is wrong.
%        Since $\frac{d \log_2 x + \log_6 (x+1)}{dx} = \frac{1}{x} + \frac{1}{x+1} < 1$ when $x \geq 7$.
%        So for $d \ge 7$,

		For $d\geq 7$, since the function $f(x) = d-2 -\log_2 d - \log_6{(d+1)}$ is monotonically increasing, and $f(7)>0$, we have
        \[ \log_2 d + \log_6{(d+1)} + 1 \leq d-1\]

%		$\epsilon \leq \frac{1}{2}^{L - \lceil d/5 \rceil}$.
%		First we need to show that for $d \leq 5$,
%		\[\frac{1}{2^{1+L-\lceil d/5\rceil}} \leq \frac{1}{2^L}\]
%
%		which is obvious because $\lceil d/5 \rceil \leq 1$,
%
%		Next we show for $d \geq 6$,
%		\[ d\cdot \left( \frac{1}{2} \right)^{d-1 + L - \lceil d/5 \rceil}  \leq \left( \frac{1}{2} \right)^L \]
%
%		Namely for $d \geq 6$,
%		\[ \log_2 d + \lceil d/5 \rceil \leq d-1\]
%
%		This is followed from that the function $f(x) = x - 2 - d/5 - \log_2 d$ is monotonically increasing, and $f(6)>0$.

		Therefore, the hypothesis for $L=k$ is verified.
		%\[\abs{P(G,e,L) - P(G,e)} \leq (\frac{1}{2})^{L+1}, \textrm{for free or dangling edge $e$}\]

		To sum up, the case of free or dangling edge and the case of normal edge together conclude the proof for our main theorem.
	%\end{proof}

\section{Counting Edge Covers}

Finally, we present the procedures for approximately counting edge covers given good estimates of the marginal probability $P(G,e)$, hence an FPTAS for the approximate counting of edge covers problem.

\begin{proposition}

    Let $Z(G) \triangleq \abs{EC(G)} \neq 0$ and $e_1,e_2,\ldots,e_m$ be an enumeration of the edges $E$ where $e_i = (u_i, v_i)$. Define $G_1 \triangleq G, G_i \triangleq G_{i-1} - e_{i-1} - u_{i-1} - v_{i-1}, 1 < i \leq m $. Then

	\[ Z(G) = \frac{1}{\prod_{i=1}^m (1 - P(G_i, e_i))} \]

\end{proposition}

\begin{proof}
	%Recall that $EC(G) \neq \emptyset$.% since the set of all edges $E$ is an edge cover.

	%Now 
    With $X \sim EC(G)$ uniformly, $\mathbb{P}(X=E)$ has two expressions,
	\begin{align*}
		\mathbb{P} (X = E) =& \frac{1}{Z(G)} \\
		\mathbb{P} (X = E) =& \prod_i \mathbb{P} \left(e_i = 1 \mid \set{e_j}_{j=1}^{i-1} = \mathbf{1} \right) \\
		=&\prod_i (1- P(G_i, e_i))
	\end{align*}

	Therefore, %we have
	\[ Z(G) = \frac{1}{\prod_{i=1}^m (1 - P(G_i, e_i))} \]
\end{proof}

We now show the main theorem of this section.
Let $Z(G, L) \triangleq \frac{1}{\prod_{i=1}^m (1 - P(G_i, e_i, L))}$ be the estimated number of edge covers given estimated $P(G_i, e_i, L)$

\begin{theorem}
	For $0< \epsilon <1$, take $L=\log_2 m + \log_2(6/ \epsilon) $,
	\[ 1- \epsilon \leq \frac{Z(G, L)}{Z(G)} \leq 1+ \epsilon\]
\end{theorem}

\begin{proof}

	\begin{align*}
		\frac{Z(G, L)}{Z(G)} &= \prod_{i=1}^m \frac{1-P(G_i, e_i)}{1-P(G_i, e_i, L)}
	\end{align*}

	By Theorem \ref{cd-main-theorem},

	\[\abs{P(G_i, e_i, L) - P(G_i,e_i)} \leq \frac{\epsilon}{4m}\]

	Recall that $1-P(G_i,e_i, L) \geq \frac{1}{2}$,
	\[ \frac{\abs{P(G_i, e_i, L) - P(G_i,e_i)}}{1 - P(G_i,e_i, L)} \leq \frac{\epsilon}{2m}\]
	
	Namely $\forall i$,
	\[ \left( 1 - \frac{\epsilon}{2m} \right) \leq \frac{1-P(G_i, e_i)}{1 - P(G_i,e_i, L)} \leq \left( 1 + \frac{\epsilon}{2m} \right)\]

	So we have
	\[ \left( 1 - \frac{\epsilon}{2m} \right)^m \leq \prod_{i=1}^m \frac{1-P(G_i, e_i)}{1 - P(G_i,e_i, L)} \leq \left( 1 + \frac{\epsilon}{2m} \right)^m\]
	\[ 1- \epsilon \leq \frac{Z(G, L)}{Z(G)} \leq 1+ \epsilon\]

\end{proof}

To sum up, since $Z(G, L)$ involves $m$ calls to $P(G,e,L)$, so run $Z(G, L)$ with $L = \log_2 m + \log_2(6/ \epsilon)$, is an FPTAS for counting edge covers with overall running time $O(m \cdot n^2 \cdot ( m\cdot \frac{1}{\epsilon})^ {\log_2 6} )$.

\section{Open Problems}
We have presented an FPTAS for approximately counting the number of edge covers for any graph. Similarly as the counting weighted independent sets with fugacity parameter $\lambda$, a natural question to ask is whether there is also an FPTAS for approximately counting weighted edge covers, or formally, is there an FPTAS to approximate the following partition function $Z_G(\lambda)$:
\[Z_G(\lambda) \triangleq \sum_{X \in EC(G)} \lambda^{\abs{X}}\]

Also, will there be a phase transition as in the case of counting independent sets? Note that our current approach can be directly extended to the case where $\lambda$ is not too small (e.g. $\lambda > \frac{4}{9}$), leaving the region where $\lambda$ being small open.

As we have noted previously, an alternative view point of the edge cover problem is Rtw-Mon-CNF, hence other natural problems are:
\begin{itemize}
	\item For what integer value of $k$, counting read $k$ times monotone CNF admits an FPTAS?
	\item For counting read twice CNF (Rtw-CNF), is there an FPTAS?
\end{itemize}
We remark that Rtw-CNF admits FPRAS~\cite{TwiceSAT}, while even counting read thrice 2CNF (without the monotone restriction) is as hard as counting 2CNF (without the read restriction) and hence does not admit FPRAS unless $RP=NP$.
However to the best of our knowledge, it is still open even whether counting Rtw-3CNF admits FPTAS.
In general, it is of interest to see how far the correlation decay technique could get in designing FPTAS for counting problems.

\bibliographystyle{plain}

\bibliography{refs}
\end{document}